%% file: main.tex
\documentclass[letterpaper, 10pt, conference]{ieeeconf}

\IEEEoverridecommandlockouts

\input{headers}

\begin{document}

\title{A Sample-Based Algorithm for Approximately Testing $r$-Robustness of a Digraph}
\author{Yuhao Yi, Yuan Wang, Xingkang He, Stacy Patterson
, and Karl H. Johansson
\thanks{Y. Yi,
Y. Wang, and K. H. Johansson are with the Division of Decision and Control Systems, School of Electrical Engineering and Computer Science, KTH Royal Institute of Technology, and they are also affiliated with Digital Futures, SE-100 44 Stockholm, Sweden (email: yuhaoy@kth.se, yuanwang@kth.se, kallej@kth.se).
They are supported in part by Knut \& Alice Wallenberg foundation, and by Swedish Research Council.}
\thanks{X. He is with the Department of Electrical Engineering, University of Notre Dame, Notre Dame, IN, 46556, USA (email: xhe9@nd.edu).}
\thanks{S. Patterson is with the Department of Computer Science, Rensselaer Polytechnic Institute, Troy, NY, 12180, USA (email: sep@cs.rpi.edu).}
}

\maketitle

\input{sections/abstract}

\input{sections/intro}
\input{sections/preliminary}
\input{sections/problem}
\input{sections/alg}
\input{sections/discussion}



\input{sections/appendix}

\end{document}

%% file: headers.tex
\usepackage[algo2e,ruled,vlined, linesnumbered]{algorithm2e}
 \usepackage{color}
\usepackage{amssymb}
\usepackage{mathtools}

\usepackage{graphicx}

\newtheorem{theorem}{Theorem}
\newtheorem{lemma}{Lemma}
\newtheorem{corollary}{Corollary}
\newtheorem{problem}{Problem}

\newtheorem{definition}{Definition}
\newtheorem{assumption}{Assumption}

\newcommand{\eps}{\epsilon}
\renewcommand{\emptyset}{\varnothing}

\def\defeq{\stackrel{\mathrm{def}}{=}}

\def\prob#1#2{\mathbb{P}_{#1}\left[ #2 \right]}

\def\expec#1#2{{\mathbb{E}}_{#1}\left[ #2 \right]}

\newcommand{\calN}{\mathcal{N}}

\def\inN#1{\mathcal{N}^{\downarrow}_{#1}}

\def\sizeof#1{\left|#1  \right|}

%% file: sections/abstract.tex
\begin{abstract}
One of the intensely studied concepts of network robustness is $r$-robustness, which is a network topology property quantified by an integer $r$. It is required by mean subsequence reduced (MSR) algorithms and their variants to achieve resilient consensus. However,  determining $r$-robustness is intractable for large networks. In this paper,
we propose a sample-based algorithm to approximately test $r$-robustness of a digraph with $n$ vertices and $m$ edges. 
For a digraph with a moderate assumption on the minimum in-degree, and an error parameter $0<\eps\leq 1$, the proposed algorithm distinguishes $(r+\eps n)$-robust graphs from graphs which are not $r$-robust with probability $(1-\delta)$.  
Our algorithm runs in $\exp(O((\ln{\frac{1}{\eps\delta}})/\eps^2))\cdot m$ time. 
%
%
%
 The running time is linear in the number of edges if $\eps$ is a constant.

\end{abstract}

%% file: sections/intro.tex
\section{Introduction}
Consensus is the cornerstone of cooperative distributed systems as a mechanism to share information among agents. Due to its wide applications, the safety aspects of the problem have been considered intensively. One of the essential problems is to design consensus algorithms that tolerate a locally or globally bounded number of faulty agents or adversaries. In this context, consensus is achieved if the honest agents agree on a value which is justified by their initial values. These algorithms are also called resilient consensus algorithms. 

The history of distributed systems and multi-agent systems has witnessed the development of a whole spectrum of consensus algorithms against adversaries: the seminal paper~\cite{PSL80} and its explanatory version~\cite{LSP82} which study the binary consensus; the paper addressing incomplete networks~\cite{Dol82} for the binary case; more recent development on scalar consensus in incomplete networks~\cite{LZKS13,VTL12}; and multi-agent vector consensus~\cite{VG13,TV13,Vai14}. The applications of fault tolerant consensus algorithms are beyond enumeration. Examples of recent applications include distributed optimization~\cite{SV16,SG18}, rendezvous of robots~\cite{PH17}, hypothesis testing~\cite{SV19}, and distributed estimation~\cite{MS19,AY21}.

Some examples in \cite{LZKS13,ZFS15} show that network connectivity is insufficient for resilient consensus. Therefore the new notion of \emph{$r$-robustness} has been proposed~\cite{LZKS13} and used as sufficient conditions for many algorithms to achieve consensus among honest agents. For example, the W-MSR~\cite{LZKS13}, SW-MSR~\cite{SPSCK17}, and DP-MSR~\cite{DI17} algorithms. $r$-robustness imposes connectivity constraints on vertex set pairs of the network. Loosely speaking, for the MSR algorithms in an $r$-robust network, each honest vertex updates its state while ignoring at most $\lfloor(r-1)/2\rfloor$ smallest and largest values from its neighbors. Then, if each honest vertex has at most $\lfloor (r-1)/2\rfloor$ malicious in-neighbors, the asymptotic resilient consensus is achieved. 

Despite the fact that many resilient consensus algorithms require an $r$-robust network, 
determining if a network is $r$-robust or not has been proven to be {coNP}-complete~\cite{ZFS15}. Therefore, no polynomial time algorithm exists for the problem unless {P}={NP}. A known algorithm which solves the problem for arbitrary digraphs is proposed in~\cite{LK13}. It enumerates all subset pairs of the digraph, which has a running time exponential in the number of vertices $n$. Since then, efforts have been made to either improve the efficiency of the algorithm, or circumvent the problem. A notable work is the recent paper~\cite{UP20}, in which the problems are formulated as integer linear programs (ILP). ILP solvers are used to improve the speed of searching. The reformulation brings practical improvement but does not give provable improvement on complexity. To bypass the problem, network construction methods are investigated to grow a network with given $r$~\cite{LZKS13,GPK17,GSK20}. Estimation of $r$ is also studied for special classes of networks, such as random networks~\cite{ZFS15} and random interdependent networks~\cite{SPS17}. In these special networks, $r$ is bounded by spectral and structural properties of the network and can be efficiently estimated. However, we note that these 
estimations are not necessarily tight in arbitrary networks. To the best of our knowledge, no existing work has been done to rigorously study the approximation of $r$ in arbitrary digraphs.


The main contribution of this paper is an algorithm for approximately testing $r$-robustness with provable guarantees. 
By setting an error bound, we study the problem of distinguishing $(r+\eps n)$-robust networks from networks that are not $r$-robust. 
We devise a randomized (Monte Carlo) algorithm that solves the problem with probability $(1-\delta)$ for an error parameter $\eps>0$, and a digraph satisfying a moderate assumption about the minimum degree. 
We prove the performance guarantee of the proposed algorithm and show the tradeoff between precision and running time.

Our algorithm is based on random sampling of vertices and has a Enforce-and-Test flavor that is seen in graph property testing~\cite{GGR98,Ron10}.
Random sampling has been shown to be successful in property testing~\cite{GGR98} and the design of approximation algorithms~\cite{AKK99} for dense graphs. We extend the technique to approximately determining $r$-robustness in digraphs, which is a new type of problem compared with its traditional applications.



\subsubsection*{Outline} The remainder of the paper is structured as follows. In Section~\ref{sec:prel} we introduce some basic definitions. In Section~\ref{sec:formulate} we give the definition of the considered problem. In Section~\ref{sec:algo} we describe the proposed algorithm, which is analyzed in Section~\ref{sec:analy}. Some discussion is given in Section~\ref{sec:disc}, followed by the conclusion. Discussion about practical implementation and numerical examples are shown in the appendix.


%% file: sections/preliminary.tex
\section{Preliminaries}
\label{sec:prel}
\subsection{Concepts and Notations}
A directed graph (digraph) $G$ is defined as a pair $(V,E)$, where $V$ and $E$ are the vertex set and the edge set. We let $\sizeof{V} = n$ and $\sizeof{E} = m$. We let $e = (u,v)$ be the directed edge from vertex $v$ to vertex $u$. We denote by $\inN{u}$ the set of in-neighbors of vertex $u$ and $\sizeof{\inN{u}}$ the in-degree of $u$. For a subset of vertices $V'\subset V$, we define the subgraph supported on $V'$ as $G[V']= (V', E')$ where $E' = \{(u,v)\in E : u,v\in V'\}$. 
Undirected graphs are viewed as bidirectional digraphs.

\begin{definition}[$r$-reachable set~\cite{LZKS13}]
Given a digraph $G=(V,E)$, a nonempty set $S\subset V$, an integer $r\geq 0$, $S$ is an $r$-reachable set if there exists a vertex $u\in S$ satisfying $\sizeof{\inN{u}\backslash S}\geq r$.
\end{definition}

\begin{definition}[$r$-robustness~\cite{LZKS13}]
A digraph is $r$-robust, if for every pair of nonempty, disjoint $A\subset V$ and $B\subset V$, at least one of $A$ and $B$ is $r$-reachable.
\end{definition}

\subsection{$r$-robust Graph and the Condition for Resilient Consensus}
We recall the condition for resilient consensus in a time-invariant synchronous network~\cite{LZKS13}. Each honest vertex in the network updates its value using a W-MSR algorithm. Each malicious vertex is allowed to send arbitrary but the same value to its out neighbors in each time step. A set of malicious vertices is said to be \emph{$F$-locally bounded} if any honest vertex in the network has at most $F$ malicious in-neighbors. $r$-robustness is the key to attain a guarantee for a consensus among honest vertices. It has been shown that in a synchronous system with a $\lfloor (r-1)/2 \rfloor$-locally bounded malicious set, the honest vertices in a $r$-robust network eventually agree on a value in the convex hull of their initial values. Then we say that the system facilitates resilient asymptotic consensus~\cite{LZKS13}. 

%% file: sections/problem.tex
\section{Problem Formulation}
\label{sec:formulate}
In this section we formulate the problem that we consider. Recall that exactly determining the robustness of a graph is {coNP}-complete. In this paper we consider the following approximation problem 
to tradeoff precision for improvement in running time.




\begin{problem}
\label{prob:test}
Given a digraph $G=(V,E)$, two integers $r>0$, $0\leq \Delta \leq n$, find an algorithm which 1) certifies $r$-robustness if $G$ is $(r+\Delta)$-robust; 2) refutes $(r+\Delta)$-robustness if $G$ is NOT $r$-robust.
\end{problem}

If $\Delta=0$, Problem~\ref{prob:test} recovers the decision problem of determining whether or not a given digraph is $r$-robust~\cite{LK13,ZFS15,UP20}.



We let the algorithm output accept to certify $r$-robustness, and output reject to refute $(r+\Delta)$-robustness.
An algorithm solves Problem~\ref{prob:test} if the map between input and output satisfies Fig.~\ref{fig:map}. We note that the algorithm can output either accept or reject for instances that are $r$-robust but not $(r+\Delta)$-robust by definition of Problem~\ref{prob:test}. An (additive) approximation algorithm with parameters $r$ and $\Delta$ is not required to distinguish the instances in this category from instances in the other two categories. This is the limitation of the approximation approach.


\begin{figure}[htbp]
    \centering
    \includegraphics[width=.85\linewidth]{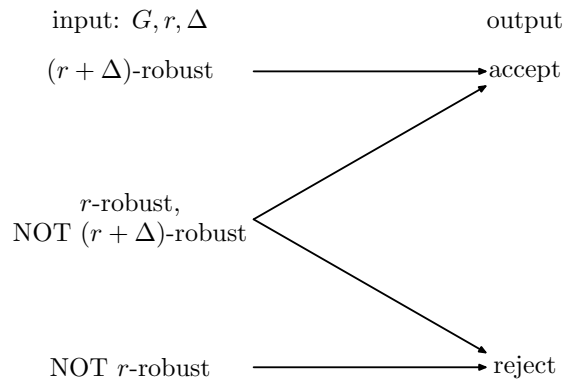}
    \caption{The map between input and output given by an algorithm that solves Problem~\ref{prob:test}.}
    \label{fig:map}
\end{figure}

An algorithm that solves Problem~\ref{prob:test} is a pessimistic testing algorithm for $r$-robustness in the sense that it produces no false answers but false negatives: it may reject $r$-robust networks but never accepts networks which are not $r$-robust. However it is not overly pessimistic since the rejected instances are not robust for a larger number $(r+\Delta)$. Alternatively, an algorithm that solves Problem~\ref{prob:test} can also be viewed as an optimistic testing algorithm for $(r+\Delta)$-robustness. Then the algorithm produces false positive answers but no false negative answers: it may accept networks that are not $(r+\Delta)$-robust but never rejects $(r+\Delta)$-robust networks. However it is not overly optimistic since the accepted instances are guaranteed to be robust for a smaller number $r$. 


\subsubsection*{A Monte Carlo Algorithm}
In this paper, we devise a Monte Carlo (MC) algorithm that behaves as follows:
 1) certifies $r$-robustness, if $G$ is $(r+\Delta)$-robust; 2) refutes $(r+\Delta)$-robustness \emph{with probability} at least $(1-\delta)$, if $G$ is not $r$-robust.
 
 In practice, we let $\delta = 1/3$. Then for a network that is not $r$-robust, it is rejected with probability at least $2/3$ in each independent run of the algorithm. We can amplify the probability of rejecting the instance to at least $(1-\sigma)$ for any $0<\sigma<1$ by running the algorithm $\lceil \frac{\ln{(1/\sigma)}}{\ln{3}}\rceil$ times\footnote{The probability that all runs fail is at most $(1/3)^{\lceil\frac{\ln{(1/\sigma)}}{\ln{3}}\rceil}\leq \sigma$. Therefore the instance is rejected in at least one run with probability at least $(1-\sigma)$.}. If the network is rejected in any one run, then it is rejected by the MC algorithm. On the other hand, networks that are $(r+\Delta)$-robust are always accepted by the MC algorithm.

%% file: sections/alg.tex
\section{Algorithm}
\label{sec:algo}
In this section we introduce an assumption for the minimum in-degree of the graph and discuss the implication of the assumption. Then we describe the approximation algorithm and provide its performance guarantee.

\subsection{An Assumption for the Minimum In-Degree}
We make an assumption for the minimum in-degree of the digraph. We show that the condition in the assumption is easy to check and is indispensable for a network to tolerate a naive attack that searches for a vertex with small in-degree and tampers with the value of half of its neighbors. Results without the assumption are left for discussion in Section~\ref{sec:woinD}.
\begin{assumption}
\label{assum:dmin}
The minimum in-degree $d_{\min}$ of the digraph $G$, defined as $d_{\min}\defeq \min_{v\in V}\{\sizeof{\inN{v}}\}$, is greater than $2r+\Delta$.
\end{assumption}
By checking $\sizeof{\inN{v}}$ for all $v\in V$, an algorithm with running time\footnote{For two positive functions $f$ and $g$ of the variable $n$, we denote $f=O(g)$ if there exist constants $n_0>0$ and $c>0$, such that for all $n>n_0$,   $f\leq c \cdot g$. We denote $f=\Omega(g)$ if $g=O(f)$.} $O(m)$ determines whether or not Assumption~\ref{assum:dmin} holds.

The robustness of a network can be interpreted as the minimum cost that a computationally unconstrained attacker has to pay to drive the system to undesired states. If we also consider the complexity of the problem, a computationally efficient strategy with a slightly larger cost could be in favor of the attacker. We show that if Assumption~\ref{assum:dmin} does not hold, it only takes $O(m)$ running time for an attacker to find an attack strategy with a reasonable cost.
\begin{lemma}
\label{lemma:degree}
Given $G$, $r$, and $\eps$, if Assumption~\ref{assum:dmin} does not hold, there exists an $O(m)$-time algorithm $\rm{ExamDegree}$ which, by checking the degree of each vertex $u\in V$,  finds a partition\footnote{A $3$-tuple $(X,Y,Z)$ is said to be a partition of the graph $G=(V,E)$ if $\min\{\sizeof{X},\sizeof{Y}\}\geq 1$, $X\cap Y = \emptyset$, $X\cap Z = \emptyset$, $Y\cap Z = \emptyset$, and  $X\cup Y \cup Z=V$. We also denote $(X,Y,\emptyset)$ as $(X,Y)$.} $(\{v\},V\backslash \{v\})$ of $G$ such that $\{v\}$ and $V\backslash \{v\}$ are not $(2r+\Delta)$-reachable.
\end{lemma}
The proof of Lemma~\ref{lemma:degree} is given in Appendix~\ref{appdxB}.

If Assumption~\ref{assum:dmin} is violated, then $\rm{ExamDegree}$ returns a vertex $u$ with minimum in-degree. For the W-MSR algorithm discussed in~\cite{LZKS13}, by attacking at least half of the in-neighbors of $u$, an attacker is able to prevent the honest vertex $u$ from reaching resilient asymptotic consensus. In particular, the vertex $u$ cannot remove all malicious messages without separating itself from all other honest vertices.

Given Lemma~\ref{lemma:degree}, we argue that to prevent a naive attack, it is necessary for the defender (or system designer) to ensure that Assumption~\ref{assum:dmin} is satisfied. In the remainder of the paper we will assume that Assumption~\ref{assum:dmin} holds. We revisit this issue and discuss arbitrary digraphs in Section~\ref{sec:woinD}.

\subsection{Reachability of Small Subsets}


Under Assumption~\ref{assum:dmin}, we propose a randomized algorithm to solve Problem~\ref{prob:test} with time complexity $\exp(\tilde{O}(1/\eps^2))m$. We will use the following concept of robustness in our analysis.
\begin{definition}[$\beta$-close $r$-robustness]
\label{close.def}
For a given $\beta \in [1/n,1]$, a digraph $G=(V,E)$ is $\beta$-close to $r$-robustness , denoted $r_\beta$-robustness, if for every pair of nonempty, disjoint $A\subset V$ and $B\subset V$ with $\min\{\sizeof{A},\sizeof{B}\}\geq \beta n$, at least one of $A$ and $B$ is $r$-reachable. 
\end{definition}
We note that the definition itself provides a weaker concept for network robustness. 

To simplify notations we let $\eps \defeq  \Delta/n$. 
If Assumption~\ref{assum:dmin} holds, we attain the following result:
\begin{lemma}
\label{lemma:smallReach}
For a graph $G$ that satisfies Assumption~\ref{assum:dmin}, the graph is $r$-robust if and only if it is $r_{\eps}$-robust, where $\eps \defeq \Delta/n$.
\end{lemma}
The proof of Lemma~\ref{lemma:smallReach} is deferred to Appendix~\ref{appdxB}.

Lemma~\ref{lemma:smallReach} shows that under Assumption~\ref{assum:dmin}, sets with sizes less than $\eps n$ are always $r$-reachable. Therefore to solve Problem~\ref{prob:test}, it suffices to approximately test $r_\eps$-robustness of the network, by examining partitions $(A,B,C)$ in which $\min\{\sizeof{A}, \sizeof{B}\} \geq \eps n$. In particular, an algorithm solves Problem~\ref{prob:test} if it rejects instances which are NOT $r_\eps$-robust, and accepts instances which are $(r+\Delta)_\eps$-robust, where $\eps=\Delta/n$.

\subsection{Algorithm Outline}
Before diving into the details, we describe the overall process of the algorithm. Our algorithm is based on vertex sampling. For any fixed partition of the network, random sampling provides statistical information for the partition. If the network is not $r_{\eps}$-robust, we seek to construct a partition $(A,B,C)$ that violates $r$-robustness. We can only reconstruct it approximately. We prove that if there exists a partition $(A,B,C)$ in which $\min\{\sizeof{A},\sizeof{B}\}\geq \eps n$, and none of $A$ and $B$ is $r$-reachable, then the proposed algorithm finds a partition $(A', B', C')$ in which none of $A'$ and $B'$ is $(r+\eps n)$-reachable, with probability\footnote{The probability can be further amplified by repetition, as we have explained in Section~\ref{sec:formulate}.} at least $(1-\delta)$. 
The number of vertices that need to be sampled is a function of the parameters $\eps$ and $\delta$, but independent of $n$, if $\eps$ and $\delta$ are positive constants. The function will be specified later in the paper.

The algorithm first randomly samples a set $U$ of vertices from the graph. The size of $U$ should be sufficiently large~\footnote{A lower bound of $\sizeof{U}$ will be later given as a function of $\eps$ and $\delta$.} for estimating the number of vertices in $\inN{v}$ that are also in subsets $A\cup C$ and $B\cup C$. For a partition $(A,B,C)$ where $A$ and $B$ are not $r$-reachable, and $\min\{\sizeof{A},\sizeof{B}\}\geq \eps n$, then with high probability there exists a partition $\pi(U)= (U_A, U_B, U_C)$ of $U$, such that $U_A\subset A$. $U_B\subset B$, $U_C\subset C$. In addition, we can estimate $\sizeof{\inN{v}\cap(A\cup C)}$ and $\sizeof{\inN{v}\cap(B\cup C)}$ using $\sizeof{\inN{v}\cap(U_A\cup U_C)}$ and $\sizeof{\inN{v}\cap(U_B\cup U_C)}$ for all $v\in (V\backslash U)$. Then there are constraints based on the sizes of the intersections that help us assign the rest of the vertices to their corresponding subsets. Specifically, if a vertex $v$ is estimated to have a large number of in-neighbors in $(A\cup C)$, it is not likely that $v$ belongs to $B$; if a vertex $v$ is estimated to have a large number of in-neighbors in $B\cup C$, it is not likely that $v$ belongs to $A$. By utilizing these constraints we attain a partition $(A',B',C')$. Let $\Gamma_{A'}$ (resp. $\Gamma_{B'}$) be the set of vertices in $A'$ (resp. $B'$) such that each vertex in $\Gamma_{A'}$ (resp. $\Gamma_{B'}$) has the number of in-neighbors from $B'\cup C'$ (resp. $A'\cup C'$) greater or equal to a threshold of $r+O(\eps n)$. The attained partition $(A',B',C')$ is constructed such that $\sizeof{\Gamma_{A'}}+\sizeof{\Gamma_{B'}}$ is $O(\eps n)$ with probability $(1-\delta)$. 
Then we run one pass of updates to correct the assignments of the misclassified vertices. In this pass at most $O(\eps n)$ vertices are moved between $A'$, $B'$, and $C'$, therefore the pass does not change the number of neighbors of any vertices in $A'$, $B'$, or $C'$ by more than $O(\eps n)$. Then we attain a partition $(A',B',C')$ in which both $A'$ and $B'$ are not $r+O(\eps n)$ reachable.

Since we do not know which partition $\pi(U)$ of the sampled vertices $U$ corresponds to the partition $(A,B,C)$ that violates $r_\eps$-robustness, we simply try all partitions of the sampled vertices.

\subsection{The Sample-Based Algorithm}
The algorithm ${\rm{SampledRbstTst}}(G, \eps, \delta, r)$ outputs accept  if the network is $(r+\eps n)$-robust. It outputs reject with probability $(1-\delta)$ if the network is not $r$-robust. The algorithm samples a set $U$ of $t(\eps, \delta)$ vertices and examine all partitions of the set $U$. For each partition $(U_A, U_B, U_C)$ of $U$, the algorithm calls $3$ subroutines $\rm{Restrict}$, $\rm{Move}$, and $\rm{TestReach}$. The algorithm $\rm{Restrict}$ is a one-pass algorithm that takes all vertices $V\backslash U$ and add the vertices to the partition $(A',B',C')$. The algorithm $\rm{Move}$ is a one-pass algorithm that refines $(A',B',C')$ based on the partition returned by $\rm{Restrict}$. The algorithm $\rm{TestReach}$ checks if a refutation of $(r+\eps n)$-robustness is found.  

The key to the approximation is the $\rm{Restrict}$ subroutine. The algorithm takes as input the graph $G$ and all the parameters from $\rm{SampledRbstTst}$, and a partition $(U_A, U_B, U_C)$ of the sampled vertices $U$. $A'$, $B'$, and $C'$ are initialized to be equal to $U_A$, $U_B$, and $U_C$ respectively. Then the subroutine designates the rest of the vertices to one of $A'$, $B'$, and $C'$. We will analyze the algorithm in the next section.


\begin{algorithm2e}[t]
\SetAlgoLined
\caption{SampledRbstTst($G$, $\eps$, $\delta$, $r$)
}
\label{alg:testRbst}
\SetKwInOut{Input}{Input}\SetKwInOut{Output}{Output}
\Input{$G=(V,E)$, an error bound $\eps>0$, \\an error probability $\delta$, a parameter $r$}
\Output{TestRslt: accept (a certificate for $r$-robustness) or reject (a refutation of $(r+\eps n)$-robustness)}
$p\gets r/n$\;
Sample a set $U$ of size $t(\eps, \delta)$ uniformly at random;\\
TestRslt$\gets$accept\;
\For{ each $3$-partition $\pi(U)=(U_A, U_B, U_C)$ where $U_A\neq \emptyset$ and $U_B\neq \emptyset$}{
\# Algorithm~\ref{alg:Restrict}
$(A',B',C')\gets \text{Restrict}(G, U_A, U_B, U_C, p, \eps, t)$;
\\
\# Algorithm~\ref{alg:Move}
$(A',B',C')\gets \text{Move}(G, A', B', C', p , \eps)$;\\
\# Algorithm~\ref{alg:testReach}
$Rslt \gets \text{TestReach}(G, A', B', C', p, \eps)$;\\
\If{$Rslt=0$}{
TestRslt$\gets$reject\; $\textbf{Break}$\;
}
}
\end{algorithm2e}

\begin{algorithm2e}[t]
\SetAlgoLined
\caption{Restrict($G, U_A, U_B, U_C, p, \eps, t$)}
\label{alg:Restrict}
\SetKwInOut{Input}{Input}\SetKwInOut{Output}{Output}
\Input{$G, U_A, U_B, U_C, p, \eps, t$}
\Output{a partition of all vertices $\pi'(V)=(A',B',C')$}
$A'\gets U_A$, $B'\gets U_B$, $C'\gets U_C$;\\
\For{ each vertex $v\in V\backslash U$}{
\uIf{$\frac{\sizeof{\inN{v}\cap (U_A\cup U_C)}}{ t} > (p+\eps /4)  \And \frac{\sizeof{\inN{v}\cap (U_B\cup U_C)}}{ t} > (p+\eps/4)$}{
$C'\gets C'\cup \{v\}$;\\
}\uElseIf{$\frac{\sizeof{\inN{v}\cap (U_A\cup U_C)}}{ t} > (p+\eps /4) $}{
$A'\gets A' \cup\{v\}$;\\
}\uElseIf{$\frac{\sizeof{\inN{v}\cap (U_B\cup U_C)}}{ t} > (p+\eps /4) $}{
$B'\gets B' \cup\{v\}$;\\
}\uElse{
add $v$ to one of  $A'$, $B'$, or $C'$ arbitrarily\;
}
}
\end{algorithm2e}

\begin{algorithm2e}[t]
\SetAlgoLined
\caption{Move($G, A', B', C', p, \eps$)}
\label{alg:Move}
\SetKwInOut{Input}{Input}\SetKwInOut{Output}{Output}
\Input{$G, A', B', C', p, \eps$, as explained}
\Output{a partition of all vertices $\pi'(V)=(A',B',C')$}
$(A'',B'',C'')\gets(A',B',C')$\;
\For{ each vertex $v\in (V\backslash U)$}{
\uIf{$v\in A' \And \sizeof{\inN{v}\cap (B'\cup C')}> (pn+3\eps n/4)$}{
\uIf{$\sizeof{\inN{v}\cap (A'\cup C')}> (pn+3\eps n/4)$}{
$A''\gets A''\backslash\{v\}$, $C''\gets C''\cup \{v\}$\;
}\uElse{
$A''\gets A''\backslash\{v\}$, $B''\gets B''\cup \{v\}$
}
}\uElseIf{$v\in B' \And \sizeof{\inN{v}\cap (A'\cup C')}> (pn+3\eps n/4)$}{
\uIf{$\sizeof{\calN_v\cap (B' \cup C')}> (pn+3\eps n/4)$}{
$B''\gets B''\backslash\{v\}$, $C''\gets C''\cup \{v\}$\;
}\uElse{
$B''\gets B''\backslash\{v\}$, $A''\gets A''\cup \{v\}$
}
}
}
$(A',B',C')\gets(A'',B'',C'')$\;
\end{algorithm2e}

\begin{algorithm2e}[t]
\SetAlgoLined
\caption{TestReach($G, A', B', C', p, \eps$)}
\label{alg:testReach}
\SetKwInOut{Input}{Input}\SetKwInOut{Output}{Output}
\Input{$G, A', B', C', p, \eps$, as explained}
\Output{Rslt: $1$ if $A'$ or $B'$ are $(pn+\eps n)$-reachable, $0$ if both are not}
Rslt$\gets 0$\;
\For{ each vertex $v\in V$}{
\uIf{$v\in A' \And \sizeof{\inN{v}\cap (B'\cup C')}\geq (pn+\eps n)$}{
Rslt$\gets 1$\;
}\uElseIf{$v\in B' \And \sizeof{\inN{v}\cap (A'\cup C')}\geq (pn+\eps n)$}{
Rslt$\gets 1$\;
}
}
\end{algorithm2e}

\section{Algorithm Analysis}
\label{sec:analy}
We show a Monte Carlo algorithm that solves Problem~\ref{prob:test} with probability at least $(1-\delta)$ for any $0<\eps\leq 1$ and $\delta>0$ in $\exp(\tilde{O}(1/\eps^2))m$ running time\footnote{The notation $\tilde{O}(\cdot)$ hides factors of polynomials of $\ln{(1/(\eps\delta))}$.} . 
\begin{theorem}
\label{theorem:main}
Given a graph $G$, two integers $r>0$, $\eps \defeq \Delta/n$ $ (\eps\in(0,1])$, under Assumption~\ref{assum:dmin}, Algorithm~\ref{alg:testRbst}
\begin{enumerate}
\item certifies $r$-robustness if $G$ is $(r+\eps n)$-robust; 
\item refutes $(r+\eps n)$-robustness, with probability at least $(1-\delta)$,  if $G$ is NOT $r$-robust;
\item runs in $\exp(\tilde{O}(1/\eps^2))\cdot m$ time.
\end{enumerate}
%
\end{theorem}

Before we start proving Theorem~\ref{theorem:main}, we explain the sampling method and 
prepare lemmas that provide guarantees to the subroutines of Algorithm~\ref{alg:testRbst}. 

\subsubsection*{Sampling}

There are several ways to sample vertices from the network. One way is to assign independent Bernoulli random variables to each vertex and sample each vertex with the same probability of $O(\text{Poly}(1/\eps)) \frac{1}{n}$. The other two ways are to sample a fixed number of $t=O(\text{Poly}(1/\eps))$ vertices uniformly at random with and without replacement. 

The analysis given in this paper is based on the scheme of uniformly sampling a fixed number of vertices with replacement. Technically speaking the outcomes of the sampling procedure are multisets. The intersection of the sampled multiset $U$ and any set $S$ is defined as a new multiset with support set $\mathrm{support}(U)\cap S$ and occurence number $\phi(u)$ the same as in $U$ for any $u\in \mathrm{support}(U)\cap S$. We ignore this point in the presentation of analysis as it is treated in the literature~\cite{GGR98}.

\subsubsection*{Approximating Number of Neighbors in Subsets}

We begin by showing that $A\cup U$ and $B\cup U$ are non-empty with bounded probability:
\begin{lemma}
\label{lemma:nonempty}
If $t\geq \frac{1}{\eps}\ln{\frac{16}{\delta}}$, with probability at least $1-\delta/8$, for $A$ and $B$ with sizes $\min\{\sizeof{A}, \sizeof{B}\}\geq \eps n$, $A\cap U$ and $B\cap U$ are non-empty.
\end{lemma}
The proof of Lemma~\ref{lemma:nonempty} is given in Appendix~\ref{appdxC}.

Then we show the following result for the $\rm{Restrict}$ algorithm.

\begin{lemma}
\label{lemma:restrict}
Let $(A,B,C)$ be a partition of $G$,  $\min\{\sizeof{A} ,\sizeof{B}\}\geq \eps n$, and both $A$ and $B$ are not $r$-reachable. Let $U$ be a set of vertices sampled uniformly at random, with $\sizeof{U}\geq \frac{8}{\eps^2}\ln{\frac{32}{\eps\delta}}$. 
Let  $\pi(U)=(U_A, U_B, U_C)$ be a partition of $U$ which satisfies $U_A\subset A$, $U_B\subset B$, $U_C\subset C$. Then under Assumption~\ref{assum:dmin}, the $\rm{Restrict}$ algorithm, which takes as input $G, \pi(U), p, \eps,$ and $\delta$, outputs a partition $(A', B', C')$ which satisfies the following property with probability $(1-\delta)$:
\begin{enumerate}
    \item[(*)] Let $\Gamma_{A'}$ (resp. $\Gamma_{B'}$) be the set of vertices in $A'$ (resp. $B'$) which consists of vertices $v$ such that $\sizeof{\inN{v}\cap (V\backslash A')}\geq r+3\eps/4$ (resp. $\sizeof{\inN{v}\cap (V\backslash B')}\geq r+3\eps/4$), then $\sizeof{\Gamma_{A'}}\leq \eps n/4$ (resp. $\sizeof{\Gamma_{B'}}\leq \eps n/4$). 
\end{enumerate}
\end{lemma}
The proof of Lemma~\ref{lemma:restrict} is provided in Appendix~\ref{appdxC}.

\subsubsection*{Correcting Large Violations}
The $\rm{Move}$ algorithm updates $(A',B',C')$ with the guarantee given by the following lemma.

\begin{lemma}
\label{lemma:move}
If Assumption~\ref{assum:dmin} holds, then
given a partition $(A',B',C')$ which satisfies property (*) described in Lemma~\ref{lemma:restrict}, the $\rm{Move}$ algorithm returns an updated $(A',B',C')$ in which $A'$ and $B'$ are both not $(r+\eps n)$-reachable.
\end{lemma}
The Proof of Lemma~\ref{lemma:move} is provided in Appendix~\ref{appdxC}.

\begin{figure}[htbp]
    \centering
    \includegraphics[width=.5\linewidth]{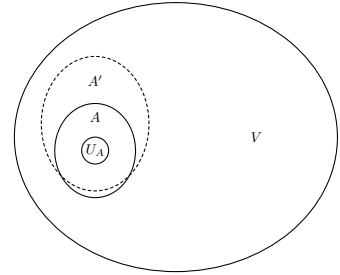}
    \caption{The relationships between the subsets $U_A$, $A$, $A'$, and $V$.}
    \label{fig:subsets}
\end{figure}

The relationships between the subsets $U_A$, $A$ and $A'$ are shown in Fig.~\ref{fig:subsets}. 
The proposed algorithm ensures that $A'$ and $B'$ are disjoint. The sets $A'\cap B$ and $B' \cap A$ can be non-empty, but include $O(\eps n)$ vertices. The size of the set $A'\backslash A$ is not necessarily small. We do not show the details of the analysis because it is irrelevant to the correctness of the algorithm. 

\subsubsection*{Complexity} Next we analyze the running time of the proposed algorithm.
\begin{lemma}
\label{lemma:time}
The running time of the \rm{SampledRbstTst} algorithm is $\exp(\tilde{O}(1/\eps^2))m$.
\end{lemma}
The proof of Lemma~\ref{lemma:time} is shown in Appendix~\ref{appdxC}.

\begin{proof}[Proof of Theorem~\ref{theorem:main}]
Suppose Assumption~\ref{assum:dmin} holds. By combining Lemmas~\ref{lemma:smallReach}, ~\ref{lemma:nonempty},~\ref{lemma:restrict}, and~\ref{lemma:move}, we know that if there exists a partition $(A,B,C)$ where $A$ and $B$ are not $r$-reachable, we obtain a partition $(A',B',C')$ in which $A'$ and $B'$ are not $(r+\Delta)$-reachable with probability at least $(1-\delta)$. Then Algorithm~\ref{alg:testReach} will return $0$ for such a partition $(A',B',C')$. Therefore Algorithm~\ref{alg:testRbst} returns reject for such instances with probability $(1-\delta)$. 

Algorithm~\ref{alg:testRbst} never rejects $G$ that is $(r+\Delta)$-robust because such partitions $(A',B',C')$ in which $A'$ and $B'$ are not $(r+\Delta)$-reachable does not exist. Therefore Algorithm~\ref{alg:testRbst} returns accept for such instances.

By combining the two properties stated above and Lemma~\ref{lemma:time}, we attain Theorem~\ref{theorem:main}.
\end{proof}

%% file: sections/discussion.tex
\section{Discussion}
\label{sec:disc}






%
%

\subsection{Estimating the Interval for Robustness}
Given Algorithm~\ref{alg:testRbst}, we can easily construct an algorithm that finds an interval of length at most $(1+\beta)\Delta$ for any constant $\beta>0$, which includes the maximal $\bar{r}$ such that the digraph is $\bar{r}$-robust. The algorithm is a modified binary search. The lower bound $\underline{\ell}$ and upper bound $\overline{\ell}$ of the interval are initialized as $0$ and $n/2$, respectively. In each round, if $\overline{\ell}-\underline{\ell}\geq (1+\beta)\Delta$, we run Algorithm~\ref{alg:testRbst} with $p = (\underline{\ell}-\Delta+ \overline{\ell})/2n$ and $\eps = \Delta/n$. If the algorithm returns accept, we let $\underline{\ell}\gets (\underline{\ell}-\Delta+ \overline{\ell})/2$, else we let $\overline{\ell} \gets (\underline{\ell}+\Delta+ \overline{\ell})/2$. We stop once  $\overline{\ell}-\underline{\ell} \leq (1+\beta)\Delta$ is satisfied. The interval $[\underline{\ell}, \overline{\ell})$ is then returned as an estimation of the interval that includes $\bar{r}$. We note that the success probability of Algorithm~\ref{alg:testRbst} needs to be amplified to guarantee the overall success probability of the binary search. We omit the analysis since it follows straightforwardly by a union bound.

\subsection{Without Assuming Minimum In-Degree}
\label{sec:woinD}
Throughout our analysis we assume that Assumption~\ref{assum:dmin} holds. On the other hand, by combining the $\rm{ExamDegree}$  algorithm and the $\rm{SampledRbstTst}$ algorithm, we attain the following corollary for an arbitrary digraph.

\begin{corollary}
Given a digraph $G=(V,E)$, two integers $r>0$, $\Delta > 0$, there exists an algorithm that runs in $\exp(\tilde{O}(1/\eps^2))m$ time which 1) outputs accept to certify $r$-robustness if $G$ is $(2r+\Delta+1)$-robust; 2) outputs reject to refute $(2r+\Delta+1)$-robustness if $G$ is not $r$-robust, with probability $(1-\delta)$.
\end{corollary}
\begin{proof}
We first run $\rm{ExamDegree}$ to calculate $d_{\min}$. If $d_{\min} \leq  2r+\Delta$, we let the algorithm output reject; if $d_{\min} > 2r+\Delta$, then we let the algorithm output the result returned by ${\rm{SampledRbstTst}}(G, \eps\defeq \Delta/n, \delta, r)$. 
\end{proof}

\subsection{Limitation of the Algorithm}
The algorithm cannot be applied to cases where $\eps=0$ regardless of the running time. In addition, if $\eps$ is $ O(n^{-1/2})$, the running time is worse than the $\exp(O(n))m$ time exact algorithm~\cite{LK13}. The gain in efficiency is attained if $\eps$ is $\Omega(n^{-\frac{1}{2}+c})$ for a constant $c>0$. If $\eps>0$ is a fixed constant (independent of $n$), the algorithm is a fixed parameter algorithm with running time linear in $m$, although it also depends on the fixed parameter $\eps$. We note that arbitrary dependency only on the parameter is allowed for fixed parameter algorithms. Similar dependency appears in property testing algorithms~\cite{GGR98} and approximation algorithms~\cite{AKK99} for dense graphs.
%



\section{Conclusion and Future Work}
\label{sec:conclude}
We have proposed an sample-based algorithm to approximately test $r$-robustness of a network. Computational complexity of the algorithm is investigated. The algorithm shows a tradeoff between precision and running time. Future work includes improving the running time of the algorithm, discussing the impact of regularity conditions in graphs, and investigating other approaches of approximation.

%% file: sections/appendix.tex
\appendix
\subsection{Probabilistic Inequalities}

We use Markov's inequality and additive Chernoff bounds.

\begin{lemma}[Markov's Inequality]
Let $X$ be a non-negative random variable, then for all $k>0$,
$\prob{}{X\geq k\cdot \expec{}{X}}\leq \frac{1}{k}$.
\end{lemma}

\begin{lemma}[Chernoff bound, \cite{Hoe63, GGR98}]
Let $X_1$, $X_2$, $\ldots$, $X_t$ be $t$ independent Bernoulli random variables where $X_i\in\{0,1\}$. Let $q\defeq (1/t)\sum_i \expec{}{X_i}$.  Then for every $\gamma \in [0,1]$, the following bounds hold:
    $\mathbb{P}[\frac{1}{t}\sum_{i=1}^t X_i > q+ \gamma]< \exp(-2\gamma^2 t),$  and 
    $\mathbb{P}[\frac{1}{t}\sum_{i=1}^t X_i < q- \gamma]< \exp(-2\gamma^2 t)$.
\end{lemma}

\subsection{Proofs from Section~\ref{sec:algo}}
\label{appdxB}
\begin{proof}[Proof of Lemma~\ref{lemma:degree}]
We note that $d_{\min}\geq r$ always holds. If Assumption~\ref{assum:dmin} does not hold, then $r \leq d_{\min} \leq (2r+\Delta)$. The problem is trivial in this case because in the partition $(\{u\}, V\backslash \{u\})$, both non-empty subsets are not $(2r+\Delta)$ reachable for a vertex $u$ with minimum degree.
An $O(m)$ time algorithm finds a vertex with minimum in-degree and the corresponding partition. We refer to the algorithm as $\rm{ExamDegree}$.
\end{proof}

\begin{proof}[Proof of Lemma~\ref{lemma:smallReach}]
We first show sufficiency. Suppose the digraph $G$ is $r_\eps$-robust. For any vertex $v$ in a non-empty subset $A\subset V$ with $\sizeof{A}<\eps n$, $\sizeof{\inN{v}\cap (V\backslash A)}\geq \sizeof{\inN{v}}-(\sizeof{A}-1)>2r+\eps n - \eps n +1 = 2r + 1 > r$. For any $A$ satisfying $\sizeof{A}\geq \eps n$, $\sizeof{\inN{v}\cap (V\backslash A)}\geq r$ holds by Definition~\ref{close.def}. Then  $A$ is $r$-reachable regardless of its size. $B$ is also $r$-reachable by similar analysis. Therefore $G$ is $r$-robust. The necessity is straightforwardly attained from definitions of $r$-robustness and $r_\beta$-robustness.
\end{proof}

\subsection{Proofs from Section~\ref{sec:analy}}
\label{appdxC}
\begin{proof}[Proof of Lemma~\ref{lemma:nonempty}]
For any $\eps >0$, we have
$\prob{U}{A\cap U = \emptyset} \leq (1-\eps)^t < \frac{\delta}{16}$, 
and 
$\prob{U}{B\cap U = \emptyset} \leq (1-\eps)^t < \frac{\delta}{16}$. 
By a union bound, the probability that $A\cap U$ or $B\cap U$ is empty is less than $\delta/8$.
\end{proof}

\begin{proof}[Proof of Lemma~\ref{lemma:restrict}]
Each vertex $u$ in $U$ is sampled uniformly at random, we let $X_{u,v}$ be the indicator variable for whether or not
$u\in (\inN{v} \cap (A\cup C))$ for a vertex $v$. Then $X_{u,v}$ is a Bernoulli random variable with 
$\prob{}{X_{u,v}=1}=\sizeof{\inN{v}\cap(A\cup C)}/n$ and $\prob{}{X_{u,v}=0}=1-\sizeof{\inN{v}\cap(A\cup C)}/n$. Then $\sizeof{\inN{v}\cap (U_A\cup U_C)}= \sum_{u\in U}X_{u,v}$ is a sum of independent random variables.

We let $t_0\defeq \frac{8}{\eps^2}\ln{\frac{32}{\eps\delta}}$, and $t \geq t_0$. We prove that if a vertex $v$ has more than $(p+\eps/2 ) n$ neighbors in $A\cup C$, the probability that it has less than $(p+\eps/4) t$ neighbors in $U_A\cup U_C$ is small. It follows that
\begin{align*}
&\prob{U}{\frac{\sizeof{\inN{v} \cap (U_A\cup U_C)}}{t} < (p+\eps/4)} \\
&\leq \prob{U}{\frac{\sizeof{\inN{v}\cap (U_A\cup U_C)}}{t} < \frac{\sizeof{\inN{v} \cap (A\cup C)}}{n}- \eps/4}\\
&< \exp(-2(\eps/4)^2 t) \leq \frac{\eps\delta}{32}\,.
\end{align*}
The first inequality is due to the fact that the latter event is a superset of the first one. The second inequality is by an additive Chernoff bound.

Similarly we prove that if a vertex $v$ has less than $pn$ neighbors in $A\cup C$, the probability that it has more than $(p+\eps/4)t$ neighbors in $U_A\cup U_C$ is small.
\begin{align*}
&\prob{U}{\frac{\sizeof{\inN{v} \cap (U_A\cup U_C)}}{t} > (p+\eps/4)}
<\frac{\eps\delta}{32}\,.
\end{align*}

Similar results also hold for the sets $B\cup C$ and $(U_B \cup U_C)$. By a union bound, the probability that any of these events happen for vertex $v$ is less than $\eps\delta/8$.

Let $X$ be a random variable defined as the number of vertices that have bad estimations of their neighborhood\footnote{We call $v$ a vertex with a bad estimation of its neighborhood if any of the $4$ events discussed above happen to $v$.}.
By the linearity of expectation, $\expec{}{X}$ is less or equal to $\eps\delta n/8$. We let $\lambda = 2/\delta$.
Then we attain
\begin{align}
\nonumber
    \prob{}{X\geq \frac{2}{\delta} \cdot \frac{\eps\delta n}{8}} \leq
    \prob{}{X\geq \frac{2}{\delta} \cdot \expec{}{X}}
    \leq \frac{\delta}{2}\,.
\end{align}
The first inequality holds because $\expec{}{X}\leq \eps\delta n/8$; the second inequality is due to Markov's inequality. Therefore the probability that at least $\eps n/4$ vertices have a bad estimation of their neighborhood is less or equal to $\delta/2$. The probability that less than $\eps n/4$ vertices have a bad estimation is greater or equal to $(1-\delta/2)$.

Next we analyze the $\rm{Restrict}$ Algorithm. With probability $(1-\delta)$, for at least $(1-\eps/4)n$ vertices the following inequalities hold:
    $\sizeof{\sizeof{\inN{v}\cap (U_A\cup U_C)}/t - \sizeof{\inN{v}\cap (A\cup C)}/{n}}\leq  \eps/4$.
We call these vertices with good estimations \emph{normal} vertices. 

Since the ordering of vertices for the loop does not affect the outcome of the subroutine, w.l.o.g., we assume that all \emph{normal} vertices are added with high priority. 

From Line 3 and 4 of Algorithm~\ref{alg:Restrict}, we know that all \emph{normal} vertices with more than $(pn+\eps n/2)$ in-neighbors in both $A\cup C$ and $B\cup C$ are added to $C'$. From Line 5 and 6, we observe that no \emph{normal} vertices with more than $pn+\eps n/2$ in-neighbors in $A\cup C$ is added to $B'$. From Line 7 and 8, we attain that no \emph{normal} vertices with more than $pn+\eps n/2$ in-neighbors in $B\cup C$ is added to $A'$. No \emph{normal} vertices will fall in to the case of Line 9 and 10, because otherwise it implies that
$\sizeof{\inN{v}\cap (A\cup C)}\leq (p+\eps/2)n$,
and $\sizeof{\inN{v}\cap (B\cup C)}\leq (p+\eps/2)n$. 
Then we attain $\sizeof{\inN{v}}\leq \sizeof{\inN{v}\cap (A\cup C)} + \sizeof{\inN{v}\cap (B\cup C)} \leq (2p+\eps)n$, which contradicts Assumption~\ref{assum:dmin}.

After we add all the \emph{normal} vertices, the rest of the vertices are added to $A'$, $B'$, and $C'$ arbitrarily.
Therefore, after the execution of the $\mathrm{Restrict}$ algorithm, we let $\Gamma^{(0)}_{A'}\defeq \{v\mid v\in A' , \sizeof{\inN{v}\cap (V\backslash A)}\geq pn+\eps n/2\}$.
 We know that (with probability $1-\delta$), $\sizeof{\Gamma^{(0)}_{A'}}\leq \eps n/4$.  In addition, $\sizeof{A\backslash A'}\leq \eps n/4$, because these vertices are not $\emph{normal}$. Then we attain that $\forall v\notin \Gamma^{(0)}_{A'}$, $\sizeof{\inN{v}\cap(V\backslash A')}<  pn+3\eps n/4$. Therefore $\sizeof{\Gamma_{A'}}\leq \sizeof{\Gamma^{(0)}_{A'}}\leq \eps n/4$. A similar result holds straightforwardly for $B'$ and $\Gamma_{B'}$.
\end{proof}

\begin{proof}[Proof of Lemma~\ref{lemma:move}]
After moving vertices with greater or equal to $pn+3\eps n/4$ in-neighbors in $V\backslash A'$ to the outside of $A'$ in the $\mathrm{Move}$ algorithm, we obtain an $A'$ in which all $v\in A'$ satisfy $\sizeof{\inN{v}\cap (V\backslash A')} < pn+\eps n$. Similarly for $v\in B'$, $\sizeof{\inN{v} \cap (V\backslash B')} < pn+\eps n$.
\end{proof}
\begin{proof}[Proof of Lemma~\ref{lemma:time}]
There are $\exp(\tilde{O}(1/\eps^2))$ partitions of the sampled set $U$ that need to be checked. The running time of \rm{Restrict}, \rm{Move}, and \rm{TestReach} for each partition is $O(m)$.
\end{proof}

\subsection{Algorithm in Practice}
\label{sec:alginP}
We discuss approaches to further improve the efficiency of the algorithm and the quality of the result. 

The first technique is to randomly assign each vertex in $U$ to one of three subsets $U_{A}$, $U_B$, $U_C$. It would take longer for the random partition to hit the correct partition $\pi^*(U)$ that aligns with the true partition $(A,B,C)$, however, by making a few changes to the algorithm we can tolerate some errors in partitioning $U$. Hitting a partition that is close enough to $\pi^*(U)$ with high probability turns out to be more efficient than enumerating all the partitions of $U$, although it does not improve the upper bound $\exp(\tilde{O}(1/\eps^2))m$. We also note that by using random partitions, the algorithm can be readily parallelized.

To use random partitions, we need to make the following changes to the algorithm. (1) we run the for loop (Line 2 to 13) of Algorithm~\ref{alg:Move} over $v\in V$ instead of $v\in (V\backslash U)$. (2) In the for loop (Line 2 to 13) of Algorithm~\ref{alg:Move}, we always check if by moving a vertex between partitions, $A$ or $B$ becomes empty. In that case we do not move the vertex.

The second technique is pruning. When we search for $\pi^*(U)$, we do not utilize the information contained in the subgraph $G[U]$. In fact by considering this information we can rule out some of the partitions of $U$. Pruning does not change the complexity of the algorithm either.

The third technique is also used to deal with the heavy dependency of the running time on $1/\eps$. Even with a constant $\eps$, the number of partitions of $U$ is $\exp(\tilde{O}(1/\eps^2))$, which can be too large from a practical point of view. In practice $t$ can be set to a value according to the computational resource available. However, it must be used with cautions when $t<\frac{8}{\eps^2}\ln{\frac{32}{\eps\delta}}$. In such a case the algorithm loses the guarantee given in the analysis. We note that even without the guarantee, the algorithm is a heuristic that can be used by the attacker to find a weak partition or by the defender to check the robustness by simulating an attacker.

Lastly, we can use heuristics to further improve the results  attained by Algorithm~\ref{alg:testRbst}. Local search can be used to reduce the reachability of $A'$ or $B'$ without increasing the reachability of the other.
Charikar’s greedy algorithm for the densest subgraph problem can also be modified to design a heuristic to post-process the results.
However we do not provide guarantees for 
these techniques.

\subsection{Numerical Examples}
\label{sec:num}
We consider example digraphs with $200$ vertices. We construct the digraphs such that their robustness is know by construction. Then we permutate the labels of the vertices and use the digraphs as input to Algorithm~\ref{alg:testRbst}. The algorithm then finds the hidden partition in which $A$ and $B$ are not $R$-reachable, and returns the smallest $R$ among all the partitions that it reconstructed. 

The digraphs are constructed as follows. We first partition the vertices into $3$ subsets $A$, $B$, and $C$. We specify the sizes of $A$ and $B$. For simplicity we let $\sizeof{A}=\sizeof{B}$, and $\sizeof{C}=200-(\sizeof{A}+\sizeof{B})$.  Each subset forms a complete digraph. Then for each vertex $u$ in $A$ (resp. $B$), we find $\bar{r}$ vertices in $B\cup C$ (resp. $A\cup C$) uniformly at random, an add edges from these vertices to $u$. Finally, we add all edges $(u,v)$ where $u\in (A\cup B)$ and $v$ is in $C$. When $\min\{\sizeof{A},\sizeof{B},\sizeof{C}\}\geq 2\bar{r}$, the constructed digraphs are $\bar{r}$-robust, but not $(\bar{r}+1)$-robust.

We note that sizes of $A$, $B$, and $C$, the number $\bar{r}$, and the ordering of the vertices are assumed to be unknown. In real scenarios one needs to choose an $r$ as part of the input. In our examples, we let $r=\bar{r}+1$ to test the effectiveness of the algorithm.  

To test the proposed algorithm, we let $\eps = 0.15$, which implies $\Delta=\eps n = 30$. Since we construct a digraph such that it is $\bar{r}$-robust but not $(\bar{r}+1)$-robust, the algorithm should return a pair of sets $A'$ and $B'$ such that both sets are not $(r+\Delta)$-robust.

\begin{figure}[htbp]
    \centering
    \includegraphics[width=.4\linewidth]{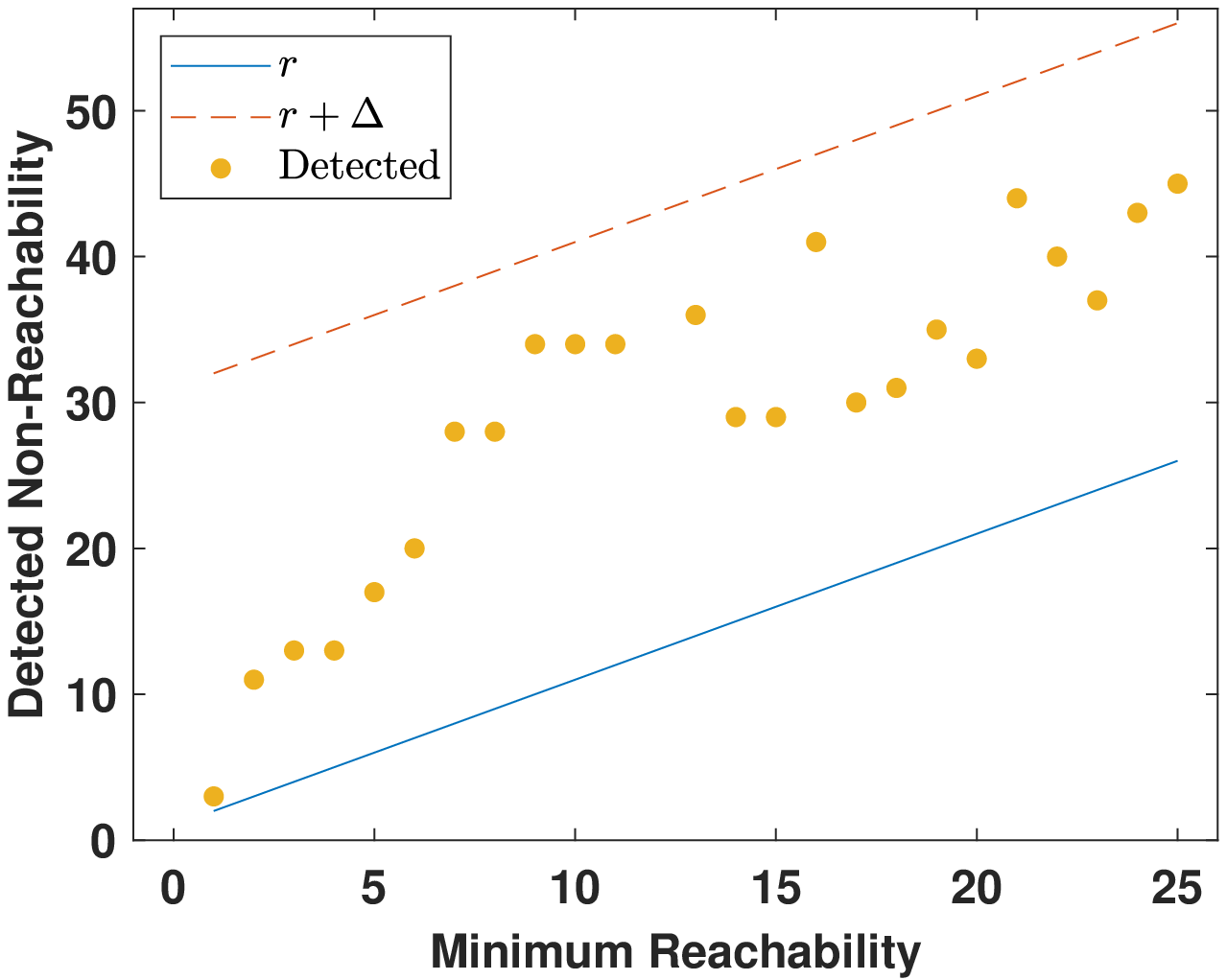}\quad
    \includegraphics[width=.4\linewidth]{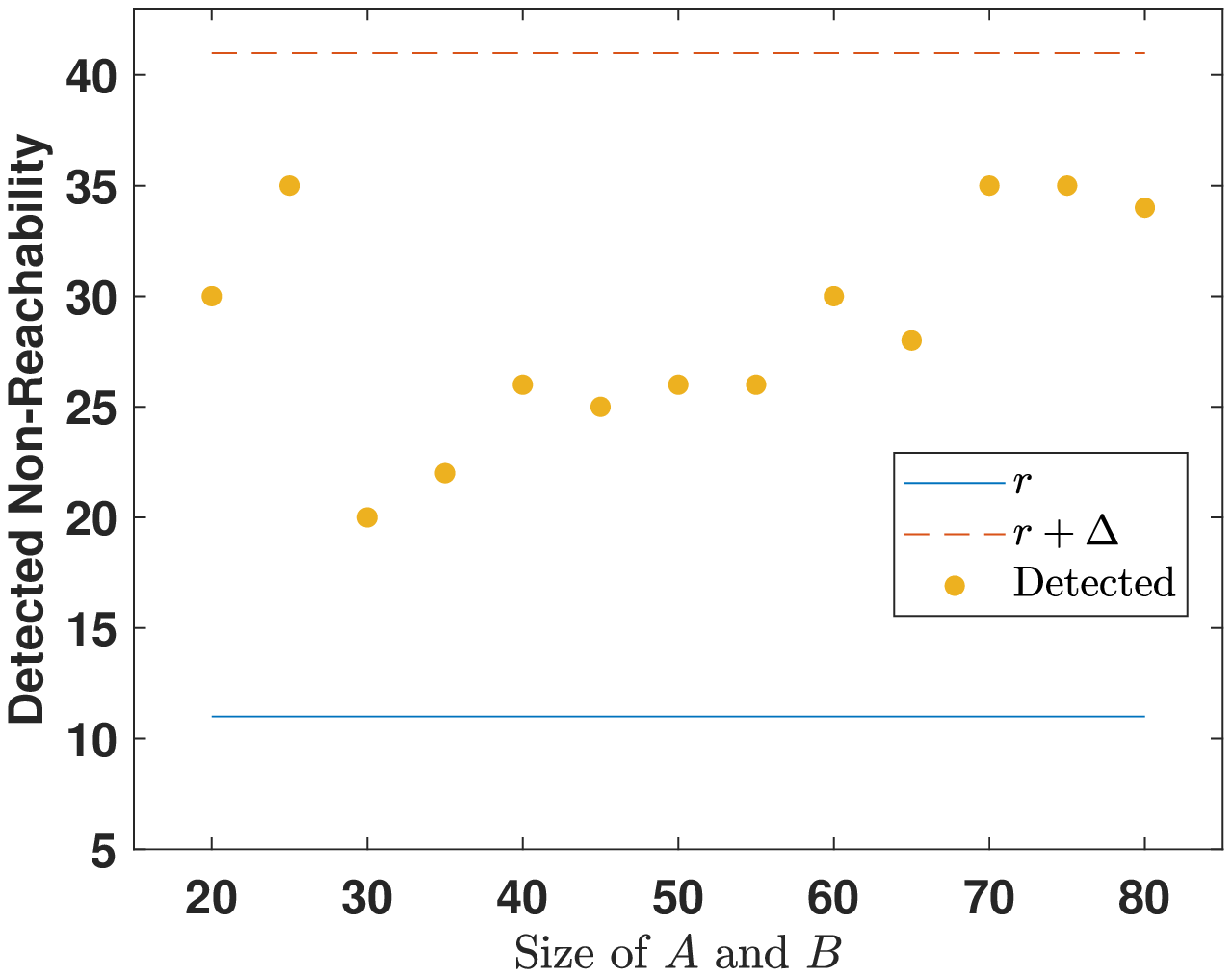}
\caption{Reachability of detected sets. Left: fixed $\sizeof{A}=\sizeof{B}=70$ and various minimum reachability $\bar{r}$; Right: fixed minimum reachablility $\bar{r}=10$ and various $\sizeof{A}=\sizeof{B}$. The solid lines show the optimum values; the dashed lines show the thresholds used by Algorithm~\ref{alg:testReach}; the data points show the values of violation detected by the algorithm.}
    \label{fig:testRbst}
\end{figure}

The algorithm is implemented using Matlab 2021b.  We run the algorithm using a single thread on a laptop computer with an Intel core i5-8365U CPU (1.6 GHz). We choose $t=9$, and run $3$ trails for each instance. We have discussed the issue of practically choosing $t$ in Section~\ref{sec:alginP}.

Fig.~\ref{fig:testRbst} shows the returned $R$ in various settings. In almost all the cases the algorithm finds subsets $A'$ and $B'$, of which none is $R$-reachable, where $R\leq r+\Delta$. The only exception is the case where $\bar{r}=12$ in the first figure. In this failed case the algorithm returns a partition in which both $A'$ and $B'$ are not $82$ reachable.

Most of the instances are rejected within $30$ seconds. The $3$ trails of the failed case take total time of less than $10^3$ seconds to finish. It is of interest to compare with exact testing algorithms given in~\cite{LK13,UP20}. However, since these comparisons are time costly for large networks, we leave them to future work. 
%
%